\documentclass{llncs}

\usepackage{amsmath}
\usepackage{graphicx}
\usepackage{caption}
\usepackage{subcaption}
\usepackage{amssymb}
\usepackage{amsfonts}
\usepackage{floatrow}

\title{On Coloring Resilient Graphs}
\author{Jeremy Kun \and Lev Reyzin}
\institute{
Department of Mathematics, Statistics, and Computer Science\\
University of Illinois at Chicago,
Chicago, IL 60607\\
\texttt{\{jkun2,lreyzin\}@math.uic.edu}
}

\newtheorem{thm}{Theorem}

\newtheorem{propn}{Proposition}
\newtheorem{obs}{Observation}
\newtheorem{prob}{Problem}

\date{}

\begin{document}
\maketitle

\begin{abstract} 
We introduce a new notion of resilience for constraint satisfaction problems,
with the goal of more precisely determining the boundary between NP-hardness
and the existence of efficient algorithms for resilient instances.  In
particular, we study $r$-resiliently $k$-colorable graphs, which are those
$k$-colorable graphs that remain $k$-colorable even after the addition of any
$r$ new edges.  We prove lower bounds on the NP-hardness of coloring
resiliently colorable graphs, and provide an algorithm that colors sufficiently
resilient graphs.  We also analyze the corresponding notion of resilience for $k$-SAT.
This notion of resilience suggests an array of open
questions for graph coloring and other combinatorial problems.
\end{abstract}

\section{Introduction and related work}

An important goal in studying NP-complete combinatorial problems is to find
precise boundaries between tractability and NP-hardness. This is often done by
adding constraints to the instances being considered until a polynomial time
algorithm is found.  For instance, while SAT is NP-hard, the restricted $2$-SAT
and XOR-SAT versions are decidable in polynomial time.  

In this paper we present a new angle for studying the boundary between
NP-hardness and tractability.  We informally define the resilience of a
constraint-based combinatorial problem and we focus on the case of resilient
graph colorability. Roughly speaking, a positive instance is resilient if it
remains a positive instance up to the addition of a constraint. For example, an
instance $G$ of Hamiltonian circuit would be ``$r$-resilient'' if $G$ has a
Hamiltonian circuit, and $G$ minus any $r$ edges \emph{still} has a Hamiltonian
circuit. In the case of coloring, we say a graph $G$ is $r$-resiliently
$k$-colorable if $G$ is $k$-colorable and will remain so even if any $r$ edges
are added. One would imagine that \emph{finding} a $k$-coloring in a very
resilient graph would be easy, as that instance is very ``far'' from being not
colorable.  And in general, one can pose the question: how resilient can
instances be and have the search problem still remain hard?\footnote{We focus
on the search versions of the problems because the decision version on
resilient instances induces the trivial ``yes'' answer.}

Most NP-hard problems have natural definitions of resiliency.  For instance,
resilient positive instances for optimization problems over graphs can be
defined as those that remain positive instances even up to the addition or
removal of any edge.  For satisfiability, we say a resilient instance is one
where variables can be ``fixed'' and the formula remains satisfiable. In
problems like set-cover, we could allow for the removal of a given number of
sets. Indeed, this can be seen as a general notion of resilience for adding
constraints in constraint satisfaction problems (CSPs), which have an extensive
literature~\cite{Kumar92}.\footnote{However, a resilience definition for
general CSPs is not immediate because the ability to add any constraint (e.g.,
the negation of an existing constraint) is too strong.}

Therefore we focus on a specific combinatorial problem, graph coloring.
Resilience is defined up to the addition of edges, and we first show that this
is an interesting notion: many famous, well studied graphs exhibit strong
resilience properties. Then, perhaps surprisingly, we prove that $3$-coloring a
$1$-resiliently 3-colorable graph is NP-hard -- that is, it is hard to color a
graph even when it is guaranteed to remain $3$-colorable under the addition of
any edge. Briefly, our reduction works by mapping positive instances of 3-SAT
to 1-resiliently 3-colorable graphs and negative instances to graphs of
chromatic number at least 4. An algorithm which can color 1-resiliently
3-colorable graphs can hence distinguish between the two. On the other hand, we
observe that $3$-resiliently $3$-colorable graphs have polynomial-time coloring
algorithms (leaving the case of 3-coloring $2$-resiliently $3$-colorable graphs
tantalizingly open). We also show that efficient algorithms exist for
$k$-coloring $\binom{k}{2}$-resiliently $k$-colorable graphs for all $k$, and
discuss the implications of our lower bounds. 

This paper is organized as follows. In the next two subsections we review the
literature on other notions of resilience and on graph coloring. In
Section~\ref{sec:resilient-sat} we characterize the resilience of boolean
satisfiability, which is used in our main theorem on 1-resilient 3-coloring. In
Section~\ref{sec:resilient-coloring-obs-bounds} we formally define the
resilient graph coloring problem and present preliminary upper and lower
bounds. In Section~\ref{sec:main-thm} we prove our main theorem, and in
Section~\ref{sec:open-problems} we discuss open problems.

\subsection{Related work on resilience}

There are related concepts of resilience in the literature. Perhaps the closest
in spirit is Bilu and Linial's notion of stability~\cite{BL12}.  Their notion
is restricted to problems over metric spaces; they argue that practical
instances often exhibit some degree of stability, which can make the problem
easier.  Their results on clustering stable instances have seen considerable
interest and have been substantially extended and
improved~\cite{ABS10,BL12,Rey11}.  Moreover, one can study TSP and other
optimization problems over metrics under the Bilu-Linial
assumption~\cite{MSSW11}.  A related notion of stability by Ackerman and
Ben-David~\cite{AckermanB09} for clustering yields efficient algorithms when
the data lies in Euclidian space.

Our notion of resilience, on the other hand, is most natural in the case when
the optimization problem has natural constraints, which can be fixed or
modified.  Our primary goal is also different -- we seek to more finely
delineate the boundary between tractability and hardness in a systematic way
across problems.

Property testing can also be viewed as involving resilience. Roughly speaking
property testing algorithms distinguish between combinatorial structures that
satisfy a property or are very far from satisfying it. These algorithms are
typically given access to a small sample depending on a parameter $\varepsilon$
alone. For graph property testing, as with resilience, the concept of being
$\varepsilon$-far from having a property involves the addition or removal of an
arbitrary set of at most $\varepsilon \binom{n}{2}$ edges from $G$. Our notion
of resilience is different in that we consider adding or removing a constant
number of constraints. More importantly, property testing is more concerned
with query complexity than with computational hardness.

\subsection{Previous work on coloring}

As our main results are on graph colorability, we review the relevant past
work. A graph $G$ is $k$-colorable if there is an assignment of $k$ distinct
colors to the vertices of $G$ so that no edge is monochromatic. Determining
whether $G$ is $k$-colorable is a classic an NP-hard problem~\cite{Karp72}.
Many attempts to simplify the problem, such as assuming planarity or bounded
degree, still result in NP-hardness~\cite{Dailey80}. A large body of work
surrounds positive and negative results for explicit families of graphs. The
list of families that are polynomial-time colorable includes triangle-free
planar graphs, perfect graphs and almost-perfect graphs, bounded tree- and
clique-width graphs, quadtrees, and various families of graphs defined by the
lack of an induced subgraph~\cite{Cai03,EBH99,HMM10,Ko03,KKTW01}.

With little progress on coloring general graphs, research has naturally turned
to approximation. In approximating the chromatic number of a general graph, the
first results were of Garey and Johnson, giving a performance guarantee of
$O(n/\log n)$ colors~\cite{Johnson74} and proving that it is NP-hard to
approximate chromatic number to within a constant factor less than
two~\cite{GJ76}. Further work improved this bound by logarithmic factors
\cite{BR90,Ha93}. In terms of lower bounds, Zuckerman~\cite{Zu07} derandomized
the PCP-based results of H{\aa}stad~\cite{Ha99} to prove the best known
approximability lower-bound to date, $O(n^{1-\varepsilon})$.

There has been much recent interest in coloring graphs which are already known
to be colorable while minimizing the number of colors used. For a 3-colorable
graph, Wigderson gave an algorithm using at most $O(n^{1/2})$
colors~\cite{Wi83}, which Blum improved to $\tilde{O}(n^{3/8})$~\cite{Blum94}.
A line of research improved this bound still further to
$o(n^{1/5})$~\cite{KawarabayashiT14}. Despite the difficulties in improving the
constant in the exponent, and as suggested by Arora~\cite{Arora11}, there is no
evidence that coloring a 3-colorable graph with as few as $O(\log n)$ colors is
hard.

On the other hand there are asymptotic and concrete lower bounds.
Khot~\cite{Khot01} proved that for sufficiently large $k$ it is NP-hard to
color a $k$-colorable graph with fewer than $k^{O(\log{k})}$ colors; this was
improved by Huang to $2^{\sqrt[3]{k}}$~\cite{Huang13}. It is also known that
for every constant $h$ there exists a sufficiently large $k$ such that coloring
a $k$-colorable graph with $hk$ colors is NP-hard~\cite{DMR06}.  In the
non-asymptotic case, Khanna, Linial, and Safra~\cite{KLS00} used the PCP
theorem to prove it is NP-hard to 4-color a 3-colorable graph, and more
generally to color a $k$ colorable graph with at most $k + 2\left \lfloor k/3
\right \rfloor - 1$ colors. Guruswami and Khanna  give an explicit reduction
for $k=3$~\cite{GuKh2000}. Assuming a variant of Khot's 2-to-1 conjecture,
Dinur et al. prove that distinguishing between chromatic number $K$ and $K'$ is
hard for constants $3 \leq K < K'$~\cite{DMR06}. This is the best conditional
lower bound we give in Section~\ref{sec:easy-bounds}, but it does not to our
knowledge imply Theorem~\ref{thm:3-1}.

Without large strides in approximate graph coloring, we need a new avenue to
approach the NP-hardness boundary. In this paper we consider the coloring
problem for a general family of graphs which we call \emph{resiliently
colorable}, in the sense that adding edges does not violate the given
colorability assumption.

\section{Resilient SAT}\label{sec:resilient-sat}

We begin by describing a resilient version of $k$-satisfiability, which is used
in proving our main result for resilient coloring in
Section~\ref{sec:main-thm}. 

\begin{prob}[resilient $k$-SAT]
A boolean formula $\varphi$ is \emph{$r$-resilient} if it is satisfiable
and remains satisfiable if any set of $r$ variables are fixed. We call
\emph{$r$-resilient $k$-SAT} the problem of finding a satisfying assignment for
an $r$-resiliently satisfiable $k$-CNF formula. Likewise, $r$-resilient CNF-SAT
is for $r$-resilient formulas in general CNF form.  
\end{prob}

The following lemma allows us to take problems that involve low (even zero)
resilience and blow them up to have large resilience and large clause size.

\begin{lemma}[blowing up]
\label{lemma:up}For all $r \geq 0$, $s \geq 1$, and $k \geq 3$, $r$-resilient
$k$-SAT reduces to $[(r+1)s-1]$-resilient $(sk)$-SAT in polynomial time.

\end{lemma}
\begin{proof}

Let $\varphi$ be an $r$-resilient $k$-SAT formula. For each $i$, let
$\varphi^i$ denote a copy of $\varphi$ with a fresh set of variables. Construct
$\psi = \bigvee_{i = 1}^s \varphi^i$. The formula $\psi$ is clearly equivalent
to $\varphi$, and by distributing the terms we can transform $\psi$ into
$(sk)$-CNF form in time $O(n^s)$. We claim that $\psi$ is $[(r+1)s -
1]$-resilient. If fewer than $(r+1)s$ variables are fixed, then by the
pigeonhole principle one of the $s$ sets of variables has at most $r$ fixed
variables. Suppose this is the set for $\varphi^1$. As $\varphi$ is
$r$-resilient, $\varphi^1$ is satisfiable and hence so is $\psi$. \hfill
$\square$ 

\end{proof}

As a consequence of the blowing up lemma for $r=0, s=2, k=3$, 1-resilient 6-SAT
is NP-hard (we reduce from this in our main coloring lower bound). Moreover, a
slight modification of the proof shows that $r$-resilient CNF-SAT is NP-hard
for all $r \geq 0$. The next lemma allows us to reduce in the other direction,
shrinking down the resilience and clause sizes.

\begin{lemma}[shrinking down]
\label{lemma:down}
Let $r \geq 1$,  $k \geq 2$, and $q = \min(r, \lfloor k/2 \rfloor)$. Then
$r$-resilient $k$-SAT reduces to $q$-resilient $(\lceil \frac{k}{2} \rceil +
1)$-SAT in polynomial time.

\end{lemma}
\begin{proof}

For ease of notation, we prove the case where $k$ is even. For a clause $C =
\bigvee_{i=1}^k x_i$, denote by $C[:k/2]$ the sub-clause consisting of the
first half of the literals of $C$, specifically $\bigvee_{i=1}^{k/2} x_i$.
Similarly denote by $C[k/2:]$ the second half of $C$. Now given a $k$-SAT
formula $\varphi = \bigwedge_{j=1}^k C_j$, we construct a $(\frac{k}{2} +
1)$-SAT formula $\psi$ by the following. For each $j$ introduce a new variable
$z_j$, and  define

\[
   \psi = \bigwedge_{j=1}^k (C_j[:k/2] \vee z_j) \wedge (C_j[k/2:] \vee
\overline{z_j}) 
\]

The formulas $\varphi$ and $\psi$ are logically equivalent, and we claim $\psi$
is $q$-resilient. Indeed, if some of the original set of variables are fixed
there is no problem, and each $z_i$ which is fixed corresponds to a choice of
whether the literal which will satisfy $C_j$ comes from the first or the second
half. Even stronger, we can arbitrarily \emph{pick} another literal in the
correct half and fix its variable so as to satisfy the clause. The
$r$-resilience of $\varphi$ guarantees the ability to do this for up to $r$ of
the $z_i$. But with the observation that there are no $l$-resilient $l$-SAT
formulas, we cannot get $k/2 + 1$ resilience when $r > k/2$,
giving the definition of $q$.
\hfill $\square$
\end{proof}

Combining the blowing up and shrinking down lemmas, we get a tidy
characterization: $r$-resilient $k$-SAT is either NP-hard or vacuously trivial.

\begin{thm}

For all $k \geq 3$, $0 \leq r < k$, $r$-resilient $k$-SAT is NP-hard.

\end{thm}
\begin{proof}

We note that increasing $k$ or decreasing $r$ (while leaving the other
parameter fixed) cannot make $r$-resilient $k$-SAT easier, so it suffices to
reduce from 3-SAT to $(k-1)$-resilient $k$-SAT for all $k \geq 3$. For any $r$
we can blow up from 3-SAT to $r$-resilient $3(r+1)$-SAT by setting $s = r+1$ in
the blowing up lemma. We want to iteratively apply the shrinking down lemma
until the clause size is $s$. If we write $s_0 = 3s$ and $s_i = \lceil s_i/2
\rceil + 1$, we would need that for some $m$, $s_m = s$ and that for each $1
\leq j < m$, the inequality $\lfloor s_j / 2 \rfloor \geq r = s-1$ holds.

Unfortunately this is not always true. For example, if $s = 10$ then $s_1 =
16$ and $16/2 < 9$, so we cannot continue. However, we can avoid this for
sufficiently large $r$ by artificially increasing $k$ after blowing up. Indeed,
we just need to find some $x \geq 0$ for which
$
   a_1 = \left \lceil \frac{3s+x}{2} \right \rceil + 1 = 2(s-1).
$
And we can pick $x = s - 6 = r - 5$, which works for all $r \geq 5$. For $r =
2,3,4$, we can check by hand that one can find an $x$ that works.\footnote{The
difference is that for $r \geq 5$ we can get what we need with only two
iterations, but for smaller $r$ we require three steps.} For $r=2$ we can
start from 2-resilient 9-SAT; for $r=3$ we can start from 16-SAT; and for $r=4$
we can start from 24-SAT.
\hfill $\square$
\end{proof}

\section{Resilient graph coloring and preliminary bounds}
\label{sec:resilient-coloring-obs-bounds}

In contrast to satisfiability, resilient graph coloring has a more interesting
hardness boundary, and it is not uncommon for graphs to have relatively high
resilience. In this section we present some preliminary bounds.

\subsection{Problem definition and remarks}

\begin{prob}[resilient coloring] 
A graph $G$ is called \emph{$r$-resiliently $k$-colorable} if $G$ remains
$k$-colorable under the addition of any set of $r$ new edges.
\end{prob}

We argue that this notion is not trivial by showing the resilience properties
of some classic graphs. These were determined by exhaustive computer search.
The Petersen graph is 2-resiliently 3-colorable. The D{\"u}rer graph is
1-resiliently 3-colorable (but not 2-resilient) and 4-resiliently 4-colorable
(but not 5-resilient). The Gr{\"o}tzsch graph is 4-resiliently 4-colorable (but
not 5-resilient). The Chv{\'a}tal graph is 3-resiliently 4-colorable (but not
4-resilient).

There are a few interesting constructions to build intuition about resilient
graphs. First, it is clear that every $k$-colorable graph is 1-resiliently
$(k+1)$-colorable (just add one new color for the additional edge), but for all
$k > 2$ there exist $k$-colorable graphs which are not 2-resiliently
$(k+1)$-colorable. Simply remove two disjoint edges from the complete graph on
$k+2$ vertices. A slight generalization of this argument provides examples of
graphs which are $\left \lfloor (k+1)/2\right \rfloor$-colorable but not $\left
\lfloor (k+1)/2 \right \rfloor$-resiliently $k$-colorable for $k \geq 3$.  On
the other hand, every $\left \lfloor (k+1)/2\right \rfloor$-colorable graph is
$(\left \lfloor (k+1)/2 \right \rfloor-1)$-resiliently $k$-colorable, since
$r$-resiliently $k$-colorable graphs are $(r+m)$-resiliently $(k+m)$-colorable
for all $m \geq 0$ (add one new color for each added edge). 

One expects high resilience in a $k$-colorable graph to reduce the number of
colors required to color it. While this may be true for super-linear
resilience, there are easy examples of $(k-1)$-resiliently $k$-colorable graphs
which are $k$-chromatic. For instance, add an isolated vertex to the complete
graph on $k$ vertices. 

\subsection{Observations}

We are primarily interested in the complexity of coloring resilient graphs, and
so we pose the question: for which values of $k,r$ does the task of
$k$-coloring an $r$-resiliently $k$-colorable graph admit an efficient
algorithm? The following observations aid us in the classification of such
pairs, which is displayed in Figure~\ref{fig:classification}.

\begin{obs}\label{obs:horizontal}
An $r$-resiliently $k$-colorable graph is $r'$-resiliently $k$-colorable for
any $r' \leq r$. Hence, if $k$-coloring is in P for $r$-resiliently
$k$-colorable graphs, then it is for $s$-resiliently $k$-colorable graphs for
all $s \geq r$.  Conversely, if $k$-coloring is NP-hard for $r$-resiliently
$k$-colorable graphs, then it is for $s$-resiliently $k$-colorable graphs for
all $s \leq r$. 
\end{obs}

Hence, in Figure~\ref{fig:classification} if a cell is in P, so are all of the
cells to its right; and if a cell is NP-hard, so are all of the cells to its
left. 

\begin{obs}\label{obs:vertical}
If $k$-coloring is in P for $r$-resiliently $k$-colorable graphs, then
$k'$-coloring $r$-resiliently $k'$-colorable graphs is in P for all $k' \leq
k$. Similarly, if $k$-coloring is in NP-hard for $r$-resiliently $k$-colorable
graphs, then $k'$-coloring is NP-hard for $r$-resiliently $k'$-colorable graphs
for all $k' \geq k$.
\end{obs}
\begin{proof}
If $G$ is $r$-resiliently $k$-colorable, then we construct $G'$ by adding a
new vertex $v$ with complete incidence to $G$. Then $G'$ is $r$-resiliently
$(k+1)$-colorable, and an algorithm to color $G'$ can be used to color $G$.
\hfill $\square$
\end{proof}

Observation~\ref{obs:vertical} yields the rule that if a cell is in P, so are
all of the cells above it; if a cell is NP-hard, so are the cells below it.
More generally, we have the following observation which allows us to apply
known bounds.

\begin{obs}\label{obs:function-bound}
If it is NP-hard to $f(k)$-color a $k$-colorable graph, then it is NP-hard to
$f(k)$-color an $(f(k)-k)$-resiliently $f(k)$-colorable graph.
\end{obs}

This observation is used in Propositions~\ref{propn:4-1}
and~\ref{propn:2-to-1-diagonal}, and follows from the fact that an
$r$-resiliently $k$-colorable graph is $(r+m)$-resiliently $(k+m)$-colorable
for all $m \geq 0$ (here $r = 0, m = f(k) - k$). 

\begin{figure}
\centering
\begin{subfigure}{.65\textwidth}
  \centering
  \scalebox{0.42}{\includegraphics{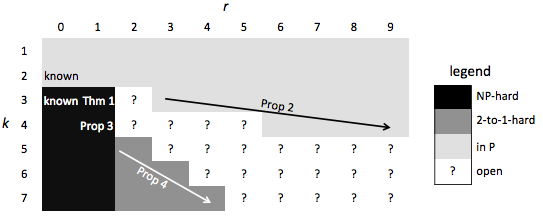}}
\end{subfigure}%
\begin{subfigure}{.35\textwidth}
  \centering
  \scalebox{0.30}{\includegraphics{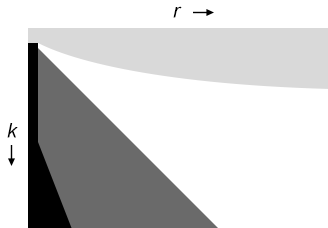}}
\end{subfigure}
\caption{The  classification of the complexity of $k$-coloring
$r$-resiliently $k$-colorable graphs. Left: the explicit classification for
small $k, r$. Right: a zoomed-out view of the same table, with the 
NP-hard (black) region added by Proposition~\ref{propn:asymptotic}.}
\label{fig:classification} \end{figure}

\subsection{Upper and lower bounds} \label{sec:easy-bounds}
In this section we provide a simple upper bound on the complexity of coloring
resilient graphs, we apply known results to show that 4-coloring a
1-resiliently 4-colorable graph is NP-hard, and we give the conditional
hardness of $k$-coloring $(k-3)$-resiliently $k$-colorable graphs for all $k
\geq 3$.  This last result follows from the work of Dinur et al., and depends a
variant of Khot's 2-to-1 conjecture~\cite{DMR06}; a problem is called
\emph{2-to-1-hard} if it is NP-hard assuming this conjecture holds. Finally,
applying the result of Huang~\cite{Huang13}, we give an asymptotic lower bound. 

All our results on coloring are displayed in Figure~\ref{fig:classification}.
To explain Figure~\ref{fig:classification} more explicitly,
Proposition~\ref{propn:k-choose-2-bound} gives an upper bound for $r=
\binom{k}{2}$, and Proposition~\ref{propn:4-1} gives hardness of the cell
$(4,1)$ and its consequences. Proposition~\ref{propn:2-to-1-diagonal} provides
the conditional lower bound, and Theorem~\ref{thm:3-1} gives the hardness of
the cell $(3,1)$. Proposition~\ref{propn:asymptotic} provides an NP-hardness
result.

\begin{propn}\label{propn:k-choose-2-bound}
There is an efficient algorithm for $k$-coloring $\binom{k}{2}$-resiliently
$k$-colorable graphs.
\end{propn}
\begin{proof}
If $G$ is $\binom{k}{2}$-resiliently $k$-colorable, then no vertex may have
degree $ \geq k$. For if $v$ is such a vertex, one may add complete incidence to
any choice of $k$ vertices in the neighborhood of $v$ to get $K_{k+1}$. Finally,
graphs with bounded degree $k-1$ are greedily $k$-colorable. 
\hfill $\square$
\end{proof}

\begin{propn}\label{propn:4-1}
4-coloring a 1-resiliently 4-colorable graph is NP-hard.
\end{propn}
\begin{proof}
It is known that 4-coloring a 3-colorable graph is NP-hard, so we may apply
Observation~\ref{obs:function-bound}. Every 3-colorable graph $G$ is
1-resiliently 4-colorable, since if we are given a proper 3-coloring of $G$ we
may use the fourth color to properly color any new edge that is added. So an
algorithm $A$ which efficiently 4-colors 1-resiliently 4-colorable graphs can
be used to 4-color a 3-colorable graph.  \hfill $\square$
\end{proof}

\begin{propn}\label{propn:2-to-1-diagonal}
For all $k \geq 3$, it is 2-to-1-hard to $k$-color a $(k-3)$-resiliently
$k$-colorable graph. 
\end{propn}
\begin{proof}
As with Proposition~\ref{propn:4-1}, we apply
Observation~\ref{obs:function-bound} to the conditional fact that it is NP-hard
to $k$-color a 3-colorable graph for $k > 3$. Such graphs are
$(k-3)$-resiliently $k$-colorable.  \hfill $\square$
\end{proof}

\begin{propn}\label{propn:asymptotic}
For sufficiently large $k$ it is NP-hard to $2^{\sqrt[3]{k}}$-color an
$r$-resiliently $2^{\sqrt[3]{k}}$-colorable graph for $r < 2^{\sqrt[3]{k}} - k$.
\end{propn}

Proposition~\ref{propn:asymptotic} comes from applying
Observation~\ref{obs:function-bound} to the lower bound of
Huang~\cite{Huang13}.  The only unexplained cell of
Figure~\ref{fig:classification} is (3,1), which we prove is NP-hard as our main
theorem in the next section.

\section{NP-hardness of 1-resilient 3-colorability}\label{sec:main-thm}


\begin{thm}\label{thm:3-1}
It is NP-hard to 3-color a 1-resiliently 3-colorable graph.
\end{thm}
\begin{proof}
We reduce 1-resilient 3-coloring from 1-resilient 6-SAT. This reduction comes
in the form of a graph which is 3-colorable if and only if the 6-SAT instance is
satisfiable, and 1-resiliently 3-colorable when the 6-SAT instance is
1-resiliently satisfiable. We use the colors
white, black, and gray.

We first describe the gadgets involved and prove their consistency (that the
6-SAT instance is satisfiable if and only if the graph is 3-colorable), and
then prove the construction is 1-resilient. Given a 6-CNF formula $\varphi = C_1
\wedge \dots \wedge C_m$ we construct a graph $G$ as follows. Start with a base
vertex $b$ which we may assume w.l.o.g.\ is always colored
gray. For each literal we construct a \emph{literal gadget} consisting of two
vertices both adjacent to $b$, as in Figure~\ref{fig:literal-gadget}. As such,
the vertices in a literal gadget may only assume the colors white and black. A
variable is interpreted as true iff both vertices in the literal
gadget have the same color. We will abbreviate this by saying a literal is
\emph{colored true} or \emph{colored false}.


\begin{figure}
\floatbox[{\capbeside\thisfloatsetup{capbesideposition={right,top},capbesidewidth=10cm}}]{figure}[\FBwidth]
{\caption{The gadget for a literal. The two single-degree vertices represent a
single literal, and are interpreted as true if they have the same color. The
base vertex is always colored gray. Note this gadget comes from
Kun~et~al.~\cite{KunPR13}.}
\label{fig:literal-gadget}}
{\scalebox{0.35}{\includegraphics{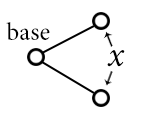}}}
\end{figure}

We connect two literal gadgets for $x, \overline{x}$ by a \emph{negation
gadget} in such a way that the gadget for $x$ is colored true if and only if
the gadget for $\overline{x}$ is colored false. The negation gadget is given in
Figure~\ref{fig:clause-and-negation-gadget}. In the diagram, the vertices
labeled 1 and 3 correspond to $x$, and those labeled 10 and 12 correspond to
$\overline{x}$. We start by showing that no proper coloring can exist if both
literal gadgets are colored true. If all four of these vertices are colored
white or all four are black, then vertices 6 and 7 must also have this color,
and so the coloring is not proper. If one pair is colored both white and the
other both black, then vertices 13 and 14 must be gray, and the coloring is
again not proper.  Next, we show that no proper coloring can exist if both
literal gadgets are colored false. First, if vertices 1 and 10 are white and
vertices 3 and 12 are black, then vertices 2 and 11 must be gray and the
coloring is not proper. If instead vertices 1 and 12 are white and vertices 3
and 10 black, then again vertices 13 and 14 must be gray. This covers all
possibilities up to symmetry.  Moreover, whenever one literal is colored true
and the other false, one can extend it to a proper 3-coloring of the whole
gadget. 

\begin{figure}
\centering
\scalebox{0.35}{\includegraphics{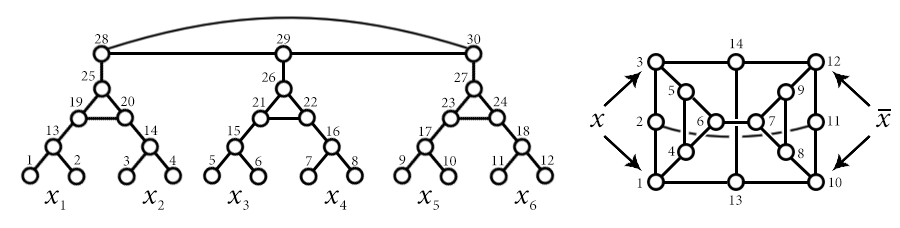}}
\caption{Left: the gadget for a clause. Right: the negation gadget ensuring two
literals assume opposite truth values.} 
\label{fig:clause-and-negation-gadget}
\end{figure}

Now suppose we have a clause involving literals, w.l.o.g., $x_1, \dots, x_6$.
We construct the \emph{clause gadget} shown in
Figure~\ref{fig:clause-and-negation-gadget}, and claim that this gadget is
3-colorable iff at least one literal is colored true. Indeed, if the literals
are all colored false, then the vertices  13 through 18 in the diagram must be
colored gray, and then the vertices 25, 26, 27 must be gray. This causes the
central triangle to use only white and black, and so it cannot be a proper
coloring. On the other hand, if some literal is colored true, we claim we can
extend to a proper coloring of the whole gadget.  Suppose w.l.o.g.\ that the
literal in question is $x_1$, and that vertices 1 and 2 both are black. Then
Figure~\ref{fig:clause-gadget-proof} shows how this extends to a proper
coloring of the entire gadget regardless of the truth assignments of the other
literals (we can always color their branches as if the literals were false). 

\begin{figure}
\vskip -.2in
\floatbox[{\capbeside\thisfloatsetup{capbesideposition={left,center},capbesidewidth=5cm}}]{figure}[\FBwidth]
{\caption{A valid coloring of the clause gadget when one variable (in this case
$x_3$) is true.}\label{fig:clause-gadget-proof}}
{\scalebox{0.35}{\includegraphics{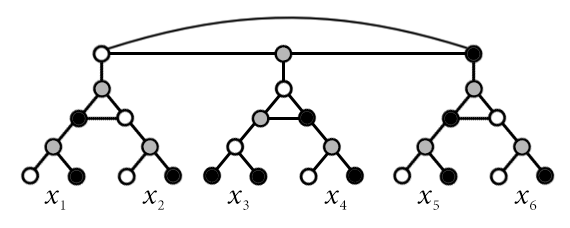}}}
\vskip -.2in
\end{figure}

It remains to show that $G$ is 1-resiliently 3-colorable when $\varphi$ is
1-resiliently satisfiable. This is because a new edge can, at worst, fix the
truth assignment (perhaps indirectly) of at most one literal. Since the
original formula $\varphi$ is 1-resiliently satisfiable, $G$ maintains
3-colorability.  Additionally, the gadgets and the representation of truth were
chosen so as to provide flexibility w.r.t.\ the chosen colors for each vertex,
so many edges will have no effect on $G$'s colorability. 

First, one can verify that the gadgets themselves are 1-resiliently
3-colorable.\footnote{These graphs are small enough to admit verification by
computer search.} We break down the analysis into eight cases based on the
endpoints of the added edge: within a single clause/negation/literal gadget,
between two distinct clause/negation/literal gadgets, between clause and
negation gadgets, and between negation and literal gadgets. We denote the added
edge by $e = (v,w)$ and call it \emph{good} if $G$ is still 3-colorable after
adding $e$. 

\emph{Literal Gadgets}. First, we argue that $e$ is good if it lies within or
across literal gadgets. Indeed, there is only one way to add an edge within a
literal gadget, and this has the effect of setting the literal to false. If $e$
lies across two gadgets then it has no effect: if $c$ is a proper coloring of
$G$ without $e$, then after adding $e$ either $c$ is still a proper coloring or
we can switch to a different representation of the truth value of $v$ or $w$ to
make $e$ properly colored (i.e. swap ``white white'' with ``black black,'' or
``white black'' with ``black white'' and recolor appropriately). 

\emph{Negation Gadgets}. Next we argue that $e$ is good if it involves a
negation gadget. Let $N$ be a negation gadget for the variable $x$. Indeed, by
1-resilience an edge within $N$ is good; $e$ only has a local effect within
negation gadgets, and it may result in fixing the truth value of $x$. Now
suppose $e$ has only one vertex $v$ in $N$. Figure~\ref{fig:coloring-negation}
shows two ways to color $N$, which together with reflections along the
horizontal axis of symmetry have the property that we may choose from at least
two colors for any vertex we wish. That is, if we are willing to fix the truth
value of $x$, then we may choose between one of two colors for $v$ so that $e$
is properly colored regardless of which color is adjacent to it.  

\begin{figure}
\vskip -.2in
\centering
\scalebox{0.35}{\includegraphics{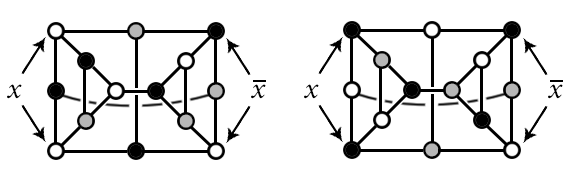}}
\caption{Two distinct ways to color a negation gadget without changing the
truth values of the literals. Only the rightmost center vertex cannot be given
a different color by a suitable switch between the two representations or a
reflection of the graph across the horizontal axis of symmetry. If the new edge
involves this vertex, we must fix the truth value appropriately.}
\label{fig:coloring-negation} \end{figure}

\emph{Clause Gadgets}. Suppose $e$ lies within a clause gadget or between two
clause gadgets. As with the negation gadget, it suffices to fix the truth value
of one variable suitably so that one may choose either of two colors for one
end of the new edge. Figure~\ref{fig:clause-clause-example} provides a detailed
illustration of one case. Here, we focus on two branches of two separate clause
gadgets, and add the new edge $e = (v,w)$. The added edge has the following
effect: if $x$ is false, then neither $y$ nor $z$ may be used to satisfy $C_2$
(as $w$ cannot be gray). This is no stronger than requiring that either $x$ be
true or $y$ and $z$ both be false, i.e., we add the clause $x \vee
(\overline{y} \wedge \overline{z})$ to $\varphi$. This clause can be satisfied
by fixing a single variable ($x$ to true), and $\varphi$ is 1-resilient, so we
can still satisfy $\varphi$ and 3-color $G$.  The other cases are analogous.

This proves that $G$ is 1-resilient when $\varphi$ is, and finishes the proof.
\hfill $\square$
\end{proof}

\begin{figure}
\vskip-.2in
\floatbox[{\capbeside\thisfloatsetup{capbesideposition={right,center},capbesidewidth=4cm}}]{figure}[\FBwidth]
{\caption{An example of an edge added between two clauses $C_1,
C_2$.}\label{fig:clause-clause-example}}
{\scalebox{0.29}{\includegraphics{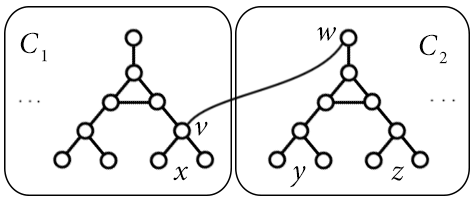}}}
\vskip-.2in
\end{figure}

\section{Discussion and open problems} \label{sec:open-problems}

The notion of resilience introduced in this paper leaves many questions
unanswered, both specific problems about graph coloring and more general
exploration of resilience in other combinatorial problems and CSPs. 

Regarding graph coloring, our paper established the fact that 1-resilience
doesn't affect the difficulty of graph coloring. However, the question of
2-resilience is open, as is establishing linear lower bounds without dependence
on the 2-to-1 conjecture. There is also  room for improvement in finding
efficient algorithms for highly-resilient instances, closing the gap between
NP-hardness and tractability.

On the general side, our framework applies to many NP-complete problems,
including Hamiltonian circuit, set cover, 3D-matching, integer LP, and many
others. Each presents its own boundary between NP-hardness and tractability,
and there are undoubtedly interesting relationships across problems.\\

\noindent \textbf{Acknowledgments}.
We thank Shai Ben-David for helpful discussions.

\bibliographystyle{plain}
\bibliography{main}

\end{document}